\documentclass[a4paper]{article}
\usepackage[affil-it]{authblk}
\usepackage[utf8]{inputenc}
\usepackage{amsmath}
\usepackage{amssymb}
\usepackage{amsthm}
\usepackage{xspace}
\usepackage{hyperref}

\newtheorem{definition}{Definition}
\newtheorem{proposition}[definition]{Proposition}

\DeclareMathOperator{\Conf}{Conf}
\let\int\relax \DeclareMathOperator{\int}{int}
\DeclareMathOperator{\Loc}{Loc}
\let\det\relax \DeclareMathOperator{\det}{det}
\DeclareMathOperator{\Sub}{Sub}

\newcommand\eg{\emph{e.g.}\xspace}
\newcommand\ie{\emph{i.e.}\xspace}

\title{Cellular Automata and Kan Extensions}
\author{Alexandre Fernandez}
\author{Luidnel Maignan}
\author{Antoine Spicher}
\affil{Univ Paris Est Creteil, LACL, 94000, Creteil, France}
\date{}

\begin{document}

\maketitle

\begin{abstract}
    In this paper, we formalize precisely the sense in which the application of a cellular automaton to partial configurations is a natural extension of its local transition function through the categorical notion of Kan extension.
    In fact, the two possible ways to do such an extension and the ingredients involved in their definition are related through Kan extensions in many ways.
    These relations provide additional links between computer science and category theory, and also give a new point of view on the famous Curtis-Hedlung theorem of cellular automata from the extended topological point of view provided by category theory.
    These links also allow to relatively easily generalize concepts pioneered by cellular automata to arbitrary kind of possibly evolving spaces.
    No prior knowledge of category theory is assumed.
\end{abstract}

\section{Introduction}

Cellular automata are usually presented either as a local behavior extended to a global and uniform one or as a continuous uniform global behavior for the appropriate topology~\cite{CeccheriniSilberstein2009}.
We offer here a third, fruitful, point of view easing many generalizations of the concepts pioneered by cellular automata, \eg via so-called global transformation~\cite{DBLP:conf/gg/MaignanS15,DBLP:conf/uc/FernandezMS19}.
The goal of this paper is not to elaborate on these generalizations but to focus on some simple foundational bridges allowing these generalizations.
In particular, we focus on Kan extensions, a categorical notion allowing, as we show here, to capture local/global descriptions~\cite{maclane2013categories}.
While categories are generalizations of monoids and posets, the case of cellular automata can be fully treated in terms of posets only.
Once the involved structures made clear via posets, the transition to category is precisely what enables the generalizations in a surprisingly smooth way as discussed in the final section.

In this paper, we recall the direct definitions of cellular automata on groups, local transition function, global transition function, shift action, and also consider the counterparts of these functions on arbitrary partial configurations.
This bigger picture allows to show that the various local/global relations between these objects are all captured by left and right Kan extensions, the latter providing a alternative definition of these objects.
The proofs are provided in detail to show how the concept can be easily manipulated once understood.
We also introduce slightly more generality that one would typically need in order to enrich the presentation of Kan extensions in a hopefully useful way.
In the final section, we comment on the link with Curtis-Hedlung theorem and discuss briefly the smooth transition to more general systems where the space itself has to evolve.

\section{Cellular Automata and Kan Extensions}

Let us give some basic definitions to fix the notations.
We also note small caveats early on, to avoid having to deal with many unrelated details at the same time in a single proof or construction latter on.

\subsection{Cellular Automata}

\begin{definition}\label{def:group}
    A \emph{group} is a set $G$ with a binary operation $\-- \cdot \-- : G \times G \to G$ which is associative, which has a neutral element $1$ and for which any $g \in G$ has inverse $g^{-1}$.
    A right action of the group on a set $X$ is a binary operation $\-- \blacktriangleleft \-- : X \times G \to X$ such that $x \blacktriangleleft 1 = x$ and $(x \blacktriangleleft g) \blacktriangleleft h = x \blacktriangleleft (g \cdot h)$.
\end{definition}

In cellular automata, the group $G$ represents the space, each element $g \in G$ being at the same time an absolute and a relative position.
This space is decorated with states that evolve through local interactions only.
The classical formal definitions go as follows and work with the entire, often infinite, space.

\begin{definition}\label{def:ca}
    A \emph{cellular automaton} on a group $G$ is given by a \emph{neighborhood} $N \subseteq G$, a finite set of \emph{states} $Q$, and a \emph{local transition function} $\delta : Q^N \to Q$.
    The elements of the set $Q^N$ are called \emph{local configurations}.
    The elements of the set $Q^G$ are called \emph{global configurations} and a right action $\-- \blacktriangleleft \-- : Q^G \times G \to Q^G$ is defined on $Q^G$ by $(c \blacktriangleleft g)(h) = c(g \cdot h)$.
    The \emph{global transition function} $\Delta : Q^G \to Q^G$ of such a cellular automaton is defined as $\Delta(c)(g) = \delta((c \blacktriangleleft g) \restriction N)$.
\end{definition}

\begin{proposition}\label{prop:action}
    The latter right action is indeed a right action.
\end{proposition}
\begin{proof}
    For any $g, h \in G$, we have
    $
            ((c \blacktriangleleft g) \blacktriangleleft h)(i)
            = (c \blacktriangleleft g)(h.i)
            = c(g\cdot h\cdot i)
            = (c \blacktriangleleft (g\cdot h))(i)
    $
    for any $i \in G$, so $((c \blacktriangleleft g) \blacktriangleleft h) = (c \blacktriangleleft (g\cdot h))$ and also $(c\blacktriangleleft 1)(i) = c(1 \cdot i) = c(i)$ as required by Definition~\ref{def:group} of right actions.
\end{proof}

This choice of definition and right notation for the so called shift action has two advantages.
Firstly, the definition of the action is a simple associativity.
Secondly, when instantiated with $G = \mathbb{Z}$ with sum, the content of $c \blacktriangleleft 5$ is the content of $c$ shifted to the left, as the symbols indicates.
Indeed, for $c' = c \blacktriangleleft 5$, $c'(-5) = c(0)$ and $c'(0) = c(5)$.

\begin{proposition}\label{prop:neighborhood}
    For all $c \in Q^G$ and $g \in G$, $\Delta(c)(g)$ is only function of $c \restriction g \cdot N$.
\end{proposition}
\begin{proof}
    Indeed, $\Delta(c)(g) = \delta((c \blacktriangleleft g) \restriction N)$ so the value is determined by $(c \blacktriangleleft g) \restriction N$.
    But for any $n \in N$, $(c \blacktriangleleft g)(n) = c(g\cdot n)$ by definition of $\blacktriangleleft$.
\end{proof}

In common cellular automata terms, this proposition means that the \emph{neighborhood of $g$} is $g \cdot N$, in this order.
Let us informally call objects of the form $c \restriction g\cdot N \in \bigcup_{g \in G}Q^{g\cdot N}$ a \emph{shifted local configuration}.
Note that, at our level of generality, two different positions $g \neq g' \in G$ might have the same neighborhood $g\cdot N = g' \cdot N$.
Although the injectivity of the function $(\--\cdot N)$ could be a useful constraints to add, which is often verified in practice, we do not impose it so the reader should keep this in mind.
A second thing to keep in mind is that we do not require here that the neighborhood should be finite.
This is done only because this property is not used in the formal development below.

\begin{proposition}\label{prop:non-inj-neighb}
The function $\-- \cdot N : G \to \mathbf{2}^G$ is not necessarily injective.
\end{proposition}
\begin{proof}
    Considering any group $G$ and $N = G$, we have $g\cdot N = G$ for any $g \in G$.
    Considering the group $G = \mathbb{Z}/2\mathbb{Z} \times \mathbb{Z}$ and $N = \{ (0, 0), (1, 0) \}$, we have $(0,0) + N = (1,0) + N = \{ (0, 0), (1, 0) \}$ because of the torsion.
\end{proof}

Because of this, it is useful to replace the shifted local configurations, \ie the union $\bigcup_{g \in G}Q^{g\cdot N}$, by the disjoint union $\bigcup_{g \in G}(\{g\} \times Q^{g \cdot N})$.
The elements of the latter are of the form $(g \in G, c \in Q^{g\cdot N})$ and keep track of the considered ``center'' of the neighborhood.
More explicitly, two elements $(g_0, c \restriction g_0 \cdot N), (g_1, c \restriction g_1 \cdot N) \in \bigcup_{g \in G}(\{g\} \times Q^{g \cdot N})$ are different as soon as $g_0 \neq g_1$ even if $g_0 \cdot N = g_1 \cdot N$.
This encodes things according to the intuition of a centered neighborhood.

\subsection{The Poset of (Partial) Configurations}

In the previous formal statements, one sees different kinds of configurations, explicitly or implicitly: global configurations $c \in Q^G$, local configurations $(c \blacktriangleleft g) \restriction N \in Q^N$, shifted local configurations $c \restriction g \cdot N \in Q^{g\cdot N}$, and their resulting ``placed states'' $(g \mapsto \Delta(c)(g)) \in Q^{\{g\}}$.
In the cellular automata literature, one often considers configurations defined on other subsets of the space, \eg finite connected subsets.
More generally, we are interested in all partial configurations $Q^S$ for arbitrary subsets $S \subseteq G$.
The restriction operation $(\-- \restriction \--)$ used many times above gives a partial ordering of these partial configurations.

\begin{definition}\label{def:conf}
    A \emph{(partial) configuration} $c$ is a partial function from $G$ to $Q$.
    Its domain of definition is denoted $|c|$ and called its \emph{support}.
    The set of all configurations is denoted $\Conf = \bigcup_{S \subseteq G}Q^S$.
    We extend the previous right action $\blacktriangleleft$ and define it to map each $c \in \Conf$ to $c \blacktriangleleft g$ having support $|c \blacktriangleleft g| = \{ h \in G \mid g \cdot h \in |c| \}$ and states $(c \blacktriangleleft g)(h) = c(g \cdot h)$.
\end{definition}

\begin{proposition}
    The latter right action is well-defined and is a right action.
\end{proposition}
\begin{proof}
    The configuration $c \blacktriangleleft g$ is well-defined on all of its support.
    Indeed for any $h \in |c \blacktriangleleft g|$, $(c \blacktriangleleft g)(h) = c(g \cdot h)$ and $g \cdot h \in |c|$ by definition of $h$.
    The right action property is verified as in the proof of Proposition~\ref{prop:action}.
\end{proof}

Let us restate Proposition~\ref{prop:neighborhood} more precisely using Definition~\ref{def:conf}.
\begin{proposition}\label{prop:neighborhood-bis}
    For all $c \in Q^G$ and $g \in G$, $(c \blacktriangleleft g) \restriction N = (c \restriction g \cdot N) \blacktriangleleft g$.
\end{proposition}
\begin{proof}
    Indeed, $|(c \restriction g \cdot N) \blacktriangleleft g| = \{ h \in G | g \cdot h \in |(c \restriction g \cdot N)| \} = \{ h \in G | g \cdot h \in g \cdot N \} = N = |(c\blacktriangleleft g) \restriction N|$.
    Also for any $n \in N$, $((c\blacktriangleleft g) \restriction N)(n) = (c\blacktriangleleft g)(n) = c(g\cdot n)$ and $((c \restriction g \cdot N) \blacktriangleleft g)(n) = (c \restriction g \cdot N)(g\cdot n) = c(g\cdot n)$.
\end{proof}

\begin{definition}
    A \emph{partial order} on a set $X$ is a binary relation ${\preceq} \subseteq X \times X$ which is reflexive, transitive, and antisymmetric.
    A set endowed with a partial order is called a \emph{partially ordered set}, or \emph{poset} for short.
\end{definition}

\begin{definition}\label{def:preceq}
    Given any two configurations $c, c' \in  \Conf$, we set $c \preceq c'$ if and only if $\forall g \in |c|, g \in |c'| \wedge c(g) = c'(g)$.
    This is read ``$c$ is a subconfiguration of $c'$'' or ``$c'$ is a superconfiguration of $c$''.
\end{definition}

\begin{proposition}\label{prop:poset-world}
    The set $\Conf$ with this binary relation is a poset.
    In this poset, the shifted local configurations $c \in \bigcup_{g\in G}Q^{g\cdot N}$ are subconfigurations of the (appropriate) global configurations $c' \in Q^G$.
    Shifted local configurations form an antichain.
    Global configurations form an antichain.
\end{proposition}
\begin{proof}
    As can be readily seen, since each global configuration restricts to many shifted local configurations, and recalling that an antichain is a subset $S$ of the poset such that neither $x \preceq x'$ nor $x' \preceq x$ hold for any two different $x,x' \in S$.
\end{proof}

\subsection{Kan Extensions (in the 2-Category of Posets)}

Given three sets $A$, $B$ and $C$ such that $A \subseteq B$, we say that a function $g : B \to C$ extends a function $f : A \to C$ if $g \restriction A = f$, or equivalently if $g \circ i = f$ where $i$ is the obvious injective function from $A$ to $B$.
For a given $f : A \to C$, there are typically many possible extensions.
Roughly speaking, Kan extensions formalizes, among many things, the mathematical practice where extensions are rarely arbitrary.
One usually chooses the ``best'' or ``most natural'' extensions.
There is therefore an implicit comparison considered between the extensions.

This is the reason why Kan extensions are formally defined at the level of 2-categories: $A$, $B$, $C$ are objects, $f$, $g$, $i$, and all (not necessarily ``most natural'') extensions are 1-morphisms between these objects, and the ``naturality'' comparison between 1-morphisms are 2-morphisms.
However, we do not need to discuss things at this level of generality here.
For our particular case, the objects are posets, the 1-morphisms are monotonic functions and the monotonic functions are compared pointwise.

\begin{definition}\label{def:monotonic}
    Given two posets $(X, \preceq_X)$ and $(Y, \preceq_Y)$, a function $f : X \to Y$ is said to be \emph{monotonic} if for all $x, x' \in X$, $x \preceq_X x'$ implies $f(x) \preceq_Y f(x')$.
\end{definition}

\begin{proposition}
    For any $g \in G$, the function $(\-- \blacktriangleleft g) : \Conf \to \Conf$ is monotonic.
\end{proposition}
\begin{proof}
    Given any $c, c' \in \Conf$ such that $c \preceq c'$, this claim is equivalent to:
    \begin{equation*}
        \begin{split}
            & (c \blacktriangleleft g) \preceq (c' \blacktriangleleft g) \text{ (by Def~\ref{def:monotonic})}\\
            \iff & \forall h \in |c \blacktriangleleft g|; h \in |c' \blacktriangleleft g| \wedge (c \blacktriangleleft g)(h) = (c' \blacktriangleleft g)(h)\text{ (Def~\ref{def:preceq})}\\
            \iff & \forall h \in G \text{ s.t. }g\cdot h \in |c|; g\cdot h \in |c'| \wedge c(g\cdot h) = c'(g\cdot h)\text{ (Def~\ref{def:conf})}\\
        \end{split}
    \end{equation*}
    which is true by the application of Definition~\ref{def:preceq} of $c \preceq c'$ on $g\cdot h$.
\end{proof}

\begin{definition}\label{def:2arrow}
    Given two posets $(X, \preceq_X)$ and $(Y, \preceq_Y)$, we define the binary relation $\-- \Rightarrow \--$ on the set of all monotonic functions from $X$ to $Y$ by $f \Rightarrow f' \iff \forall c \in \Conf, f(c) \preceq_Y f'(c)$.
\end{definition}

\begin{proposition}
    Given two posets $(X, \preceq_X)$ and $(Y, \preceq_Y)$, the set of monotonic functions between them together with this binary relation forms a poset.
\end{proposition}
\begin{proof}
    As one can easily check.
\end{proof}

\begin{definition}\label{def:kan}
    In this setting, given three posets $A$, $B$ and $C$, and three monotonic functions $i : A \to B$, $f : A \to C$ and $g: B \to C$, $g$ is said to be the \emph{left (resp. right) Kan extension of $f$ along $i$ if $g$} is the $\Rightarrow$-minimum (resp. $\Rightarrow$-maximum) element in the set of monotonic functions $\{ h : B \to C \mid f \Rightarrow h \circ i \}$ (resp. $\{ h : B \to C \mid h \circ i \Rightarrow f \}$).
\end{definition}

This concept is particularly useful because, whenever it applies, it is also a complete characterization as stated in the following proposition in the left case.
\begin{proposition}
    The left Kan extension $g$ is unique when it exists.
\end{proposition}
\begin{proof}
    It is defined as the minimum of a set, and a minimum is always unique when it exists.
\end{proof}

Another suggestive way to read the concept of Kan extensions with respect to this paper is to say that a function $g$ on a poset can be summarized into, or generated by, a part of its behavior $f$ on just a small part of the poset.
Note however that $i$ does not need to be injective in this definition.

\section{Kan Extensions in Cellular Automata}

\subsection{A First Approach To Partial Configurations}

The first, intuitive, approach is to take a configuration $c$, look for all places $g$ where the whole neighborhood $g \cdot N$ is defined and to take the local transition result of these places only.
We first give a direct formal definition, and then show that this is a left Kan extension.
This shows in particular that the global transition function is the left Kan extension of the ``fully shifted'' local transition.
The sense of ``fully shifted'' is described below and is only necessary because we restrict ourselves to posets, as discussed in the final section of this paper.

\subsubsection{A Direct Definition}

\begin{definition}\label{def:int}
    The \emph{interior} of a subset $S \subseteq G$ is $\int(S) = \{ g \in G \mid g \cdot N \subseteq S \}$.
\end{definition}
\begin{definition}\label{def:coarse-trans}
    The \emph{coarse transition function} $\Phi : \Conf \to \Conf$ is defined for all $c \in \Conf$ as $|\Phi(c)| = \int(|c|)$ and $\Phi(c)(g) = \delta((c \blacktriangleleft g) \restriction N)$.
\end{definition}

\begin{proposition}
    For any $c \in \Conf$ and $g \in G$, the statements $g \in \int(|c|)$, $g \cdot N \subseteq |c|$, and $N \subseteq |c \blacktriangleleft g|$ are equivalent.
    (So $\Phi$ is  well-defined in Definition~\ref{def:coarse-trans}.)
\end{proposition}
\begin{proof}
    The first and second statements are equivalent by Definition~\ref{def:int} of $\int$.
    The second and third statements are equivalent by Definition~\ref{def:conf} of $\blacktriangleleft$.
\end{proof}

Remember Proposition~\ref{prop:non-inj-neighb}.
If we do have injectivity of neighborhoods, we have $\int(g \cdot N) = \{g\}$.
But since we do not assume it, we only have the following.

\begin{proposition}
    Let $S \subseteq G$.
    In general, $S \subseteq \int(S\cdot N)$ but we do not necessarily have equality, even when $S = \{g\}$ for some $g \in G$.
\end{proposition}
\begin{proof}
    Consider the examples in the proof of Proposition~\ref{prop:non-inj-neighb}.
\end{proof}

Another useful remark on which we come back below is the following.

\begin{proposition}
    For any $g \in G$, any $M \subsetneq N$, and any $c \in Q^{g\cdot M}$, $|\Phi(c)| = \emptyset$.
    Also, for any $c \in Q^G$, $|\Phi(c)| = \Delta(c)$.
\end{proposition}
\begin{proof}
    By Definition~\ref{def:coarse-trans} of $\Phi$.
\end{proof}

\subsubsection{Characterization as a Left Kan Extension}

The coarse transition function $\Phi$ is defined on the set of all configurations $\Conf$ and we claim that it is generated, in the Kan extension sense, by the local transition function $\delta$ shifted everywhere.
We define the latter, with Proposition~\ref{prop:non-inj-neighb} in mind.

\begin{definition}\label{def:loc}\label{def:shifted-local-trans}\label{def:pi2-loc}
    We define $\Loc$ to be the poset $\Loc = \bigcup_{g \in G}(\{g\} \times Q^{g \cdot N})$ with trivial partial order $(g,c) \preceq (g',c') \iff (g,c) = (g',c')$.
    The ``fully shifted local transition function'' $\overline{\delta} : \Loc \to \Conf$ is defined, for any $(g, c) \in \Loc$ as $|\overline{\delta}(g,c)| = \{g\}$ and $\overline{\delta}(g,c)(g) = \delta(c \blacktriangleleft g)$.
    The second projection of $\Loc$ is the monotonic function $\pi_2 : \Loc \to \Conf$ defined as $\pi_2(g,c) = c$.
\end{definition}

\begin{proposition}
    $\Loc$ is a poset and $\overline{\delta}$ and $\pi_2$ are monotonic functions.
\end{proposition}
\begin{proof}
    Indeed, the identity relation is an order relation and any function respects the identity relation.
\end{proof}

\begin{proposition}
    The coarse transition function $\Phi$ is monotonic.
\end{proposition}
\begin{proof}
    Indeed, take $c, c' \in \Conf$ such that $c \preceq c'$.
    We want to prove that $\Phi(c) \preceq \Phi(c')$ and this is equivalent to:
    \begin{equation*}
        \begin{split}
            & \forall g \in |\Phi(c)|, g \in |\Phi(c')| \wedge \Phi(c)(g) = \Phi(c')(g) \\
            \iff & \forall g \in \int(|c|), g \in \int(|c'|) \wedge \delta(c \blacktriangleleft g \restriction N) = \delta(c' \blacktriangleleft g \restriction N) \\
            \iff & \forall g \in G \text{ s.t. } g\cdot N \subseteq |c|, g\cdot N \subseteq |c'| \wedge \delta(c \blacktriangleleft g \restriction N) = \delta(c' \blacktriangleleft g \restriction N),
        \end{split}
    \end{equation*}
    by Definition~\ref{def:conf} of $\preceq$, Definition~\ref{def:coarse-trans} of $\Phi$, and Definition~\ref{def:int} of $\int$.
    The final statement is implied by:
    \begin{equation*}
        \begin{split}
            & \forall g \in G \text{ s.t. } g\cdot N \subseteq |c|, g\cdot N \subseteq |c'| \wedge c \blacktriangleleft g \restriction N = c' \blacktriangleleft g \restriction N \\
            \iff & \forall g \in G \text{ s.t. } g\cdot N \subseteq |c|, g\cdot N \subseteq |c'| \wedge \forall n \in N, (c \blacktriangleleft g)(n) = (c' \blacktriangleleft g)(n) \\
            \iff & \forall g \in G \text{ s.t. } g\cdot N \subseteq |c|, g\cdot N \subseteq |c'| \wedge \forall n \in N, c(g \cdot n) = c'(g \cdot n),
        \end{split}
    \end{equation*}
    the last equivalence being by Definition~\ref{def:conf}.
    To prove it, take $g \in G$ such that $g \cdot N \subseteq |c|$, and $n \in N$.
    Since $c \preceq c'$ and $g\cdot n \in |c|$, we have $g \cdot n \in |c'|$, and $c(g \cdot n) = c'(g \cdot n)$ as wanted.
\end{proof}

\begin{proposition}
    $\Phi$ is the left Kan extension of $\overline{\delta}$ along $\pi_2 : \Loc \to \Conf$.
\end{proposition}
\begin{proof}
    By Definition~\ref{def:kan} of left Kan extensions, we need to prove firstly that $\Phi$ is such that $\overline{\delta} \Rightarrow \Phi \circ \pi_2$, and secondly that it is smaller than any other such monotonic functions.
    
    For the first part, $\overline{\delta} \Rightarrow \Phi \circ \pi_2$ is equivalent to:
    \begin{equation*}
        \begin{split}
            & \forall (g,c) \in \Loc, \overline{\delta}(g,c) \preceq \Phi(c) \text{ (Defs.~\ref{def:2arrow} and~\ref{def:pi2-loc} of $\Rightarrow$ and $\pi_2$)}\\
            \iff & \forall (g,c) \in \Loc, \forall h \in |\overline{\delta}(g,c)|, h \in |\Phi(c)| \wedge \overline{\delta}(g,c)(h) = \Phi(c)(h) \text{ (D.~\ref{def:preceq} of $\preceq$)}\\
            \iff & \forall (g,c) \in \Loc, g \in |\Phi(c)| \wedge \delta(c \blacktriangleleft g) = \Phi(c)(g) \text{ (Def.~\ref{def:shifted-local-trans} of $\overline{\delta}$)}\\
            \iff & \forall (g,c) \in \Loc, g \cdot N \in |c| \wedge \delta(c \blacktriangleleft g) = \delta((c \blacktriangleleft g) \restriction N) \text{ (Defs~\ref{def:int},~\ref{def:coarse-trans} of $\Phi$).}
        \end{split}
    \end{equation*}
    This last statement is true by Definition~\ref{def:loc} of $\Loc$, since $c \in Q^{g\cdot N}$.
    
    For the second part, let $F: \Conf \to \Conf$ be a monotonic function such that $\overline{\delta} \Rightarrow F \circ \pi_2$.
    We want to show that $\Phi \Rightarrow F$, which is equivalent to:
    \begin{equation*}
        \begin{split}
            & \forall c \in \Conf, \Phi(c) \preceq F(c) \text{ (Def.~\ref{def:2arrow} of $\Rightarrow$)}\\
            \iff & \forall c \in \Conf, \forall g \in |\Phi(c)|, g \in |F(c)| \wedge \Phi(c)(g) = F(c)(g) \text{ (Def.~\ref{def:preceq} of $\preceq$)}\\
            \iff & \forall c \in \Conf, \forall g \in \int(|c|), g \in |F(c)| \wedge F(c)(g) = \delta((c \blacktriangleleft g) \restriction N) \text{ (Def.~\ref{def:coarse-trans})}
        \end{split}
    \end{equation*}
    So take $c \in \Conf$ and $g \in \int(|c|)$, and consider $d_g = c \restriction g \cdot N$.
    Since $d_g \preceq c$ and $F$ is monotonic, we have $F(d_g) \preceq F(c)$.
    Moreover $\overline{\delta} \Rightarrow F \circ \pi_2$ and $(g, d_g) \in \{g\} \times Q^{g\cdot N} \subseteq \Loc$,
    so $\overline{\delta}(g, d_g) \preceq F(d_g)$ by Definitions~\ref{def:2arrow} and~\ref{def:pi2-loc} of $\Rightarrow$ and $\pi_2$.
    By transitivity $\overline{\delta}(g,d_g) \preceq F(c)$.
    By Definition~\ref{def:shifted-local-trans} of $\overline{\delta}$ and Definition~\ref{def:preceq} of $\preceq$, we obtain $g \in |F(c)|$, and $F(c)(g) =  \overline{\delta}(g, d_g)(g) = \delta((c \blacktriangleleft g) \restriction N)$, as wanted.
\end{proof}

As a sidenote, remark that in order to have the equality $\overline{\delta} = \Phi \circ \pi_2$, one needs to have the injectivity of neighborhood function.
Indeed, without injectivity, we have two different $g, g' \in G$ having the same neighborhood $M$, \ie $M = g \cdot N = g' \cdot N$.
This means that, given any local configuration $c \in Q^M$ on this neighborhood, each pair $(g,c), (g',c) \in \Loc$ have different results $\overline{\delta}(g,c) \in Q^{\{g\}}$ and
$\overline{\delta}(g',c) \in Q^{\{g'\}}$ with different support $\{g\}$ and $\{g'\}$.
However, their common projection $\pi_2(g,c)=\pi_2(g',c)=c$ have a unique result $\Phi(c)$ with a support such that $\{g,g'\} \subseteq |\Phi(c)|$.
So we have a strict comparison $\overline{\delta} \Rightarrow \Phi \circ \pi_2$.
When the neighborhood function is injective, $\pi_2$ is also injective and the previous situation can not occur so we have equality $\overline{\delta} = \Phi \circ \pi_2$.

\subsection{A Second Approach To Partial Configurations}

For some applications, the previous definitions are too naive.
For example, it is common to consider two cellular automata to be essentially the same if they generate the same global transition functions.
However, here, two such cellular automata give different coarse transition function if they have a different neighborhood.

To refine the previous definitions, a second approach is to take a configuration $c$, and look at all places for which the result is already determined by the partial data defined in $c$.
So we consider all $g \in G$ for which all completions of the data present on the defined neighborhood $g \cdot N \cap |c|$ into a configuration on the complete neighborhood $g \cdot N$ always lead to the same result by $\delta$.

\subsubsection{A Direct Definition}

\begin{definition}\label{def:restr-neighb}
    For any $c \in \Conf$ and $g \in G$, let $c_g = c \restriction (g \cdot N \cap |c|)$.
\end{definition}

\begin{definition}\label{def:det}
    Given a configuration $c \in \Conf$, its \emph{determined subset} is $\det(c) = \{ g \in G \mid \exists q \in Q,\, \forall c' \in Q^{g \cdot N},\, c' \restriction |c_g| = c_g \implies \delta(c' \blacktriangleleft g) = q \}$.
    For any $g \in \det(c)$, we denote $q_{c,g} \in Q$ the unique state $q$ having the mentioned property.
\end{definition}

Note that this definition depends on the cellular automaton local transition function $\delta$ and on the data of the configuration $c$, contrary to Definition~\ref{def:int} of interior that only depends on its neighborhood $N$ and on the support of the configuration.

\begin{definition}\label{def:fine-trans}
    Given a cellular automaton, its \emph{fine transition function} $\phi : \Conf \to \Conf$ is defined as $|\phi(c)| = \det(c)$ and $\phi(c)(g) = q_{c,g}$, \ie $\phi(c)(g) = \delta(c' \blacktriangleleft g)$ for any $c' \in Q^N$ such that $c' \restriction |c_g| = c_g$.
\end{definition}

\begin{proposition}
    The fine transition function $\phi$ is well defined.
\end{proposition}
\begin{proof}
    This is the case precisely because we restrict the support of $\phi(c)$ to the determined subset of the $c$.
\end{proof}

\begin{proposition}
    Consider the constant cellular automaton $\delta(c) = q \;\forall c \in Q^N$ for a specific $q \in Q$ and regardless of the neighborhood $N$ chosen to represent it.
    We have $|\phi(c)| = G$ for any $c \in \Conf$.
\end{proposition}
\begin{proof}
    Indeed, even with no data at all, \ie for $c$ such that $|c| = \emptyset$, the result at all position is determined and is $q$.
\end{proof}

\subsubsection{Characterization as a Right Kan Extension}

As for the coarse transition function, the fine transition function $\phi$ is defined on the set of all configurations $\Conf$ and we claim that it is generated, in the Kan extension sense.
We consider two ways to generate it and start by the simplest one.
The second one is considered in the following section using sub-local configurations in order to be closer to the direct definition and to be a ``from local to global'' characterization.

\begin{proposition}\label{prop:monotonic-restr-neighb}
    For any $g \in G$, the function $\--_g : \Conf \to \Conf$ of Definition~\ref{def:restr-neighb} is monotonic.
\end{proposition}
\begin{proof}
    As one can easily check.
\end{proof}

\begin{proposition}\label{prop:monotonic-fine-trans}
    The fine transition function $\phi$ is monotonic.
\end{proposition}
\begin{proof}
    Indeed, take $c_0, c_1 \in \Conf$ such that $c_0 \preceq c_1$.
    We want to prove that $\phi(c_0) \preceq \phi(c_1)$ and this is equivalent to:
    \begin{equation*}
        \begin{split}
            & \forall g \in |\phi(c_0)|, g \in |\phi(c_1)| \wedge \phi(c_0)(g) = \phi(c_1)(g) \text{ (Def~\ref{def:preceq} of $\preceq$)}\\
            \iff & \forall g \in \det(c_0), g \in \det(c_1) \wedge q_{c_0,g} = q_{c_1,g} \text{ (Def~\ref{def:fine-trans} of $\phi$)}
        \end{split}
    \end{equation*}
    Take $g \in \det(c_0)$. We want to prove that $g \in \det(c_1)$, which means by Definition~\ref{def:det} of $\det(c_1)$:
    \begin{equation*}
        \exists q \in Q,\, \forall c_2 \in Q^{g \cdot N},\, c_2 \restriction |(c_{1})_{g}| = (c_{1})_{g} \implies \delta(c_2 \blacktriangleleft g) = q
    \end{equation*}
    We claim that the property is verified with $q = q_{c_0,g}$.
    Indeed, take any $c_2 \in Q^{g\cdot N}$ such that $c_2 \restriction |(c_{1})_{g}| = (c_{1})_{g}$.
    We also have that $c_2 \restriction |(c_{0})_{g}| = (c_{0})_{g}$ since
    the hypothesis $c_0 \preceq c_1$ implies $(c_{0})_{g} \preceq (c_{1})_{g}$ by Proposition~\ref{prop:monotonic-restr-neighb}.
    By Definition~\ref{def:det} of $\det(c_0)$, we obtain that $\delta(c_2 \blacktriangleleft g) = q_{c_0,g}$, so $q = q_{c_0,g}$ has the wanted property, which implies that $g \in \det(c_1)$ as wanted.
    But the above property of $q$ set it to be precisely what we denote by $q_{c_1,g}$ (Def~\ref{def:det} of $q_{c_1,g}$), so $q_{c_0,g} = q_{c_1,g}$.
\end{proof}

\begin{proposition}
    The fine transition function $\phi$ is the right Kan extension of the global transition function $\Delta$ along the inclusion $i : Q^G \to \Conf$.
\end{proposition}
\begin{proof}
    By Definition~\ref{def:kan} of right Kan extensions, we need to prove firstly that $\phi$ is such that $\phi \circ i \Rightarrow \Delta$, and secondly that it is greater than any other such monotonic functions.
    
    For the first part, we actually have $\phi \circ i = \Delta$ since for any $c \in Q^G$, $|\phi(c)| = \det(c) = G = |\Delta(c)|$ and for any $g \in G$, we have $\phi(c)(g) = q_{c,g} = \delta(c_g \blacktriangleleft g) = \delta((c \restriction g\cdot N) \blacktriangleleft g) = \delta((c \blacktriangleleft g) \restriction N) = \Delta(c)(g)$ by Defs.~\ref{def:fine-trans},~\ref{def:det},~\ref{def:restr-neighb},~\ref{def:ca} of $\phi$, $\det$, $c_g$ and $\Delta$ and Prop.~\ref{prop:neighborhood-bis}.
    
    For the second part, let $f: \Conf \to \Conf$ be a monotonic function such that $f \circ i \Rightarrow \Delta$.
    We want to show that $f \Rightarrow \phi$, which is equivalent to:
    \begin{equation*}
        \begin{split}
            & \forall c \in \Conf, f(c) \preceq \phi(c) \text{ (Def.~\ref{def:2arrow} of $\Rightarrow$)}\\
            \iff & \forall c \in \Conf, \forall g \in |f(c)|, g \in |\phi(c)| \wedge f(c)(g) = \phi(c)(g) \text{ (Def.~\ref{def:preceq} of $\preceq$)}\\
            \iff & \forall c \in \Conf, \forall g \in |f(c)|, g \in \det(c) \wedge f(c)(g) = q_{c,g} \text{ (Def.~\ref{def:shifted-sub-local-trans} of $\phi$)}\\
            \iff & \forall c \in \Conf, \forall g \in |f(c)|, \forall c' \in Q^{g \cdot N}, c \preceq c' \implies f(c)(g) = \delta(c' \blacktriangleleft g) \text{ (D.~\ref{def:det})}
        \end{split}
    \end{equation*}
    So take $c \in \Conf$ and $g \in |f(c)|$ and $c' \in Q^{g \cdot N}$ such that $c \preceq c'$.
    Consider any $c'' \in Q^G$ such that $c' \preceq c''$ (or equivalently $c'' \restriction g\cdot N = c'$).
    Since $f$ is monotonic, we have $f(c) \preceq f(c'')$, which means that $f(c)(g) = f(c'')(g)$ by Def.~\ref{def:preceq}.
    But since $f \circ i \Rightarrow \Delta$, we have $f(c)(g) = \Delta(c'')(g) = \delta((c'' \blacktriangleleft g) \restriction N)$ by Def.~\ref{def:2arrow} of $\Rightarrow$ and Def.~\ref{def:ca} of $\Delta$.
    But by Prop.~\ref{prop:neighborhood-bis}, $(c'' \blacktriangleleft g) \restriction N = (c'' \restriction g\cdot N) \blacktriangleleft g = c' \blacktriangleleft g$.
\end{proof}

\subsection{Introducing Sub-Local Configurations}

The direct definition of the fine transition function is explicitly about assigning a result for a configuration $c$ at a given $g \in G$ even when the whole neighborhood $g\cdot N$ is not complete.
By isolating these ``shifted sub-local configuration'' in the poset of configurations, we can (right-)extend the local transition to them and show that, in the same way as the coarse transition function is the left Kan extension of the local transition function, the fine transition function is the left Kan extension of the sub-local transition function.

\subsubsection{Direct Definition}

\begin{definition}\label{def:shifted-sub-local-trans}\label{def:pi2-sub}
    We define $\Sub = \bigcup_{g \in G, M \subseteq N}(\{g\} \times Q^{g \cdot M})$ with partial order defined as $(g,c) \preceq (g',c')$ if and only if $g = g'$ and $c \preceq c'$.
    The ``fully shifted sub-local transition function'' $\underline{\delta} : \Sub \to \Conf$ is defined, for any $g \in G$, any $M \subseteq N$ and any $c \in Q^{g\cdot M}$, as $|\underline{\delta}(g,c)| = \{g\} \cap \det(c)$ and, if $g \in \det(c)$, $\underline{\delta}(g,c)(g) = q_{c',g}$, \ie $\underline{\delta}(g,c)(g) = \delta(c' \blacktriangleleft g)$ for any $c' \in Q^{g\cdot N}$ such that $c = c' \restriction |c|$.
    The second projection of $\Sub$ is the function $\pi_2 : \Sub \to \Conf$ defined as $\pi_2(g,c) = c$.
\end{definition}

In this definition, a given sub-local configuration can result either in an empty configuration when the transition is not determined, or in a configuration with only singleton support when the transition is determined.

Note that for a given cellular automaton, it is possible to restrict the poset $\Sub$ to an antichain.
Indeed, any time a result is determined by a sub-local configuration $(g,c)$, all bigger sub-local configuration $(g,c')$ with $c \preceq c'$ does not contribute anything new.
We do not elaborate on this because this antichain would be different for each cellular automaton, blurring the global picture presented below.

\subsubsection{Characterization as a Right Kan Extension}

\begin{proposition}
    The fully shifted sub-local transition function $\underline{\delta}$ is monotonic
\end{proposition}
\begin{proof}
    As usual, take $(g,c), (g',c') \in \Sub$ such that $(g,c) \preceq (g',c')$.
    First note that $g = g'$ by Definition~\ref{def:shifted-sub-local-trans}.
    We want to prove that $\underline{\delta}(g,c) \preceq \underline{\delta}(g,c')$ and this is equivalent to:
    \begin{equation*}
        \begin{split}
            & \forall h \in |\underline{\delta}(c)|, h \in |\underline{\delta}(c')| \wedge \underline{\delta}(c)(h) = \underline{\delta}(c')(h) \\
            \iff & \forall h \in \{g\} \cap \det(c), h \in \{g\} \cap \det(c') \wedge q_{c,g} = q_{c',g} \\
            \iff & g \in \det(c) \implies g \in \det(c') \wedge q_{c,g} = q_{c',g'},
        \end{split}
    \end{equation*}
    by Definition~\ref{def:conf} of $\preceq$ and Definition~\ref{def:shifted-sub-local-trans} of $\underline{\delta}$.
    The end of this proof is similar to the one of Proposition~\ref{prop:monotonic-fine-trans}.
\end{proof}

\begin{proposition}
    The fully shifted sub-local transition function $\underline{\delta}$ is the right Kan extension of the fully shifted local transition function $\overline{\delta}$ along the inclusion $i : \Loc \to \Sub$.
\end{proposition}
\begin{proof}
    By Definition~\ref{def:kan} of right Kan extensions, we need to prove firstly that $\underline{\delta}$ is such that $\underline{\delta} \circ i \Rightarrow \overline{\delta}$, and secondly that it is greater than any other such monotonic functions.
    
    For the first part, $\underline{\delta} \circ i \Rightarrow \overline{\delta}$ is equivalent to:
    \begin{equation*}
        \begin{split}
            & \forall (g,c) \in \Loc, \underline{\delta}(g,c) \preceq \overline{\delta}(g,c) \text{ (Def.~\ref{def:2arrow} of $\Rightarrow$)}\\
            \iff & \forall (g,c) \in \Loc, \forall h \in |\underline{\delta}(g,c)|, h \in |\overline{\delta}(g,c)| \wedge \underline{\delta}(g,c)(h) = \overline{\delta}(g,c)(h) \text{ (D.~\ref{def:preceq} $\preceq$)}\\
            \iff & \forall (g,c) \in \Loc, g \in \det(c) \implies g \in |\overline{\delta}(g,c)| \wedge q_{c,g} = \overline{\delta}(g,c)(g) \text{ (Def.~\ref{def:shifted-sub-local-trans} of $\underline{\delta}$)}\\
            \iff & \forall (g,c) \in \Loc, g \in \det(c) \implies g \in \{g\} \wedge q_{c,g} = \delta(c \blacktriangleleft g) \text{ (Def~\ref{def:shifted-local-trans} of $\overline{\delta}$).}
        \end{split}
    \end{equation*}
    This last statement is true by Def.~\ref{def:det} of $q_{c,g}$.
    
    For the second part, let $f: \Sub \to \Conf$ be a monotonic function such that $f \circ i \Rightarrow \overline{\delta}$.
    We want to show that $f \Rightarrow \underline{\delta}$, which is equivalent to:
    \begin{equation*}
        \begin{split}
            & \forall (g,c) \in \Sub, f(g,c) \preceq \underline{\delta}(g,c) \text{ (Def.~\ref{def:2arrow} of $\Rightarrow$)}\\
            \iff & \forall (g,c) \in \Sub, \forall h \in |f(g,c)|, h \in |\underline{\delta}(g,c)| \wedge f(g,c)(h) = \underline{\delta}(g,c)(h) \text{ (Def.~\ref{def:preceq})}\\
            \iff & \forall (g,c) \in \Sub, \forall h \in |f(g,c)|, h \in \{g\} \cap \det(c) \wedge f(g,c)(h) = q_{c,g} \text{ (Def.~\ref{def:shifted-sub-local-trans})}
        \end{split}
    \end{equation*}
    So take $(g,c) \in \Sub$ and $h \in |f(g,c)|$.
    Consider any $c' \in \Loc$ such that $c \preceq c'$.
    Since $f$ is monotonic, we have $f(g,c) \preceq f(g,c')$, which means that $h \in |f(g,c')|$ and $f(c)(h) = f(c')(h)$ by Def.~\ref{def:preceq}.
    But since $f \circ i \Rightarrow \overline{\delta}$, we have $h \in |\overline{\delta}(g,c')| = \{g\}$ and $f(g,c')(h) = \overline{\delta}(g,c')(h) = \delta(c' \blacktriangleleft g)$ by Def.~\ref{def:2arrow} of $\Rightarrow$ and Def.~\ref{def:shifted-local-trans} of $\overline{\delta}$.
    Since this is true for any $c'$, this establishes exactly the defining property of $\det(c)$ by Def.~\ref{def:det}.
\end{proof}

\subsubsection{The Second Approach as a Left Kan Extension}

\begin{proposition}
    The projection function $\pi_2 : \Sub \to \Conf$ is monotonic.
\end{proposition}
\begin{proof}
    As can be readily checked in Definition~\ref{def:shifted-sub-local-trans}
\end{proof}

\begin{proposition}
    $\phi$ is the left Kan extension of $\underline{\delta}$ along $\pi_2 : \Sub \to \Conf$.
\end{proposition}
\begin{proof}
    By Definition~\ref{def:kan} of left Kan extensions, we need to prove firstly that $\phi$ is such that $\underline{\delta} \Rightarrow \phi \circ \pi_2$, and secondly that it is smaller than any other such monotonic functions.
    
    For the first part, $\underline{\delta} \Rightarrow \phi \circ \pi_2$ is equivalent to:
    \begin{equation*}
        \begin{split}
            & \forall (g,c) \in \Sub, \underline{\delta}(g,c) \preceq \phi(c) \text{ (Defs.~\ref{def:2arrow} and~\ref{def:pi2-sub} of $\Rightarrow$ and $\pi_2$)}\\
            \iff & \forall (g,c) \in \Sub, \forall h \in |\underline{\delta}(g,c)|, h \in |\phi(c)| \wedge \underline{\delta}(g,c)(h) = \phi(c)(h) \text{ (Def~\ref{def:preceq} of $\preceq$)}\\
            \iff & \forall (g,c) \in \Sub, g \in \det(c) \implies g \in |\phi(c)| \wedge q_{c,g} = \phi(c)(g) \text{ (Def.~\ref{def:shifted-sub-local-trans} of $\underline{\delta}$)}\\
            \iff & \forall (g,c) \in \Sub, g \in \det(c) \implies g \in \det(c) \wedge q_{c,g} = q_{c,g} \text{ (Def~\ref{def:fine-trans} of $\phi$)},
        \end{split}
    \end{equation*}
    a most trivial statement.
    
    For the second part, let $f: \Conf \to \Conf$ be a monotonic function such that $\underline{\delta} \Rightarrow f \circ \pi_2$.
    We want to show that $\phi \Rightarrow f$, which is equivalent to:
    \begin{equation*}
        \begin{split}
            & \forall c \in \Conf, \phi(c) \preceq f(c) \text{ (Def.~\ref{def:2arrow} of $\Rightarrow$)}\\
            \iff & \forall c \in \Conf, \forall g \in |\phi(c)|, g \in |f(c)| \wedge \phi(c)(g) = f(c)(g) \text{ (Def.~\ref{def:preceq} of $\preceq$)}\\
            \iff & \forall c \in \Conf, \forall g \in \det(c), g \in |f(c)| \wedge q_{c,g} = f(c)(g) \text{ (Def.~\ref{def:fine-trans} of $\phi$)}
        \end{split}
    \end{equation*}
    So take $c \in \Conf$ and $g \in \det(c)$.
    Since $c_g \preceq c$ (Def~\ref{def:restr-neighb}) and $f$ is monotonic, we have $f(c_g) \preceq f(c)$.
    Moreover $\underline{\delta} \Rightarrow f \circ \pi_2$ and $(g, c_g)  \in \Sub$ so $\underline{\delta}(g, c_g) \preceq f(c_g)$ by Definitions~\ref{def:2arrow} and~\ref{def:pi2-loc} of $\Rightarrow$ and $\pi_2$.
    By transitivity $\underline{\delta}(g, c_g) \preceq f(c)$.
    By Definition~\ref{def:shifted-sub-local-trans} of $\underline{\delta}$ and Definition~\ref{def:preceq} of $\preceq$, we therefore have $g \in |f(c)|$, and $f(c)(g) =  \underline{\delta}(g, d_g)(g) = q_{c,g}$ as wanted.
\end{proof}

\section{Final Discussion}

There are additional simple structural facts to note about the monotonic functions considered.
The first one is that the shift action on partial configurations, as given in Definition~\ref{def:conf}, is the right Kan extension of the shift action on global configurations, as given in Definition~\ref{def:ca}.
Another one is that $\Phi \Rightarrow \phi$, hence the names of these transition functions, coarse and fine.
In fact, any monotonic function $f : \Conf \to \Conf$ such that $\overline{\delta} \Rightarrow f \circ \pi_2$ is necessarily such that $\Phi \Rightarrow f \Rightarrow \phi$.
This shows, in some sense, the efficiency of the simple constraints of monotonicity and $\overline{\delta} \Rightarrow f \circ \pi_2$.

In the formal development presented here, we explicitly ``copy'' a single local behavior $\delta$ on all $g \in G$ to obtain $\overline{\delta}$ and work with it.
It is readily possible to put a different behavior on each $g \in G$, with no real modification to the proofs.
The statements are therefore valid for non-uniform cellular automata and automata networks.
As mentioned in the beginning, we did not even used the finiteness of the neighborhood either.
At this point, the reader might have the feeling that these results are not really about cellular automata, and there are at least three answers to that.
The first answer is that one could easily impose the shift and simultaneously prevent the use of a highly redundant ``fully shifted local transition function'', but this requires using a category of configuration instead of a poset of configurations.
The latter is very similar to the poset, except that the yes/no question ``is this configuration a subconfiguration of this other one ?'' is replaced by the open-ended question ``where does this configuration appear in this other one ?''~\cite{DBLP:conf/gg/MaignanS15}.
The goal of this paper is indeed to introduce the concepts needed for this other point of view, among many others.
The second answer is that the proofs are more about the decomposition/composition process involved in the local/global definition of cellular automata.
Because of the simplicity of cellular spaces, groups, the description is very simple to make ``directly''.
In other situations, a Kan extension presentation can be the most effective way to describe the spatial extensions/restriction, for example when the space is an evolving graph~\cite{DBLP:conf/gg/MaignanS15,DBLP:conf/uc/FernandezMS19}.
The third answer is that, with small modifications, this result is closely related to the Curtis-Hedlung theorem.
Indeed, if one restores the finite neighborhoods and finite states constraints, one can see that the poset of finite support configurations is a ``generating'' part of the poset of open subsets of the product topology.
In this case, the fine transition function $\phi$ can be viewed as encoding an important part of the topological behavior of the global transition function $\Delta$~\cite{CeccheriniSilberstein2009}.

To finish, let us mention an important aspect of the Kan extensions considered here and in other papers of the author~\cite{DBLP:conf/gg/MaignanS15,DBLP:conf/uc/FernandezMS19}.
They have the property to be \emph{pointwise}.
Intuitively, this means that they can be computed ``algorithmically'' using simple building blocks.
This formulation in terms of building blocks is completely equivalent and is the one used in the other papers, firstly because it is via these building blocks that the authors discovered these links between spatially-extended dynamical systems and category theory, and secondly because this formulation is closer to the software implementations of the considered models.
In fact, it is possible to have an implementation completely generic over the particular kind of space considered, \eg evolving graphs of any sort, evolving higher-order structures such as abstract cell~\cite{DBLP:conf/gg/MaignanS15}, evolving strings such as Lindenmayer systems~\cite{DBLP:conf/uc/FernandezMS19}, or Cayley graphs as considered here.

\bibliography{main}

\end{document}